\documentclass[12pt]{article}
\usepackage[margin=1in]{geometry}
\usepackage[normalem]{ulem}
\usepackage{amsthm, amsmath, amssymb, bbm, caption, enumerate, hyperref, marvosym, setspace, times, titlesec, natbib}
\usepackage[pdftex] {graphicx}
\usepackage[usenames,dvipsnames]{xcolor}
\hypersetup{colorlinks, urlcolor=black, linkcolor=blue, citecolor=blue}
\usepackage[ruled]{algorithm2e}
\newtheorem{theorem}{Theorem}

\newtheorem{lemma}[theorem]{Lemma}

\newtheorem{condition}{Condition}

\def\T{{ \top }}

\def\bZ{{Z}}
\def\bbeta{{\beta}}

\definecolor{vaiomod}{rgb}{0.53, 0.56, 0.72}

\title{\textbf{Event History Analysis of Dynamic Communication Networks}}
\author{Tony Sit$^{*}$, Zhiliang Ying$^{\dagger}$ and Yi Yu$^{\ddagger, }$\textsuperscript{\Letter}\\~\\
\small{$^*$\textit{Department of Statistics, The Chinese University of Hong Kong}} \\
\small{$^\dagger$\textit{Department of Statistics, Columbia University}} \\
\small{$^\ddagger$\textit{School of Mathematics, University of Bristol}} \\
\scriptsize{tonysit@sta.cuhk.edu.hk \quad zying@stat.columbia.edu  \quad y.yu@bristol.ac.uk}}
\date{}
\begin{document}
\maketitle 

\abstract

Statistical analysis on networks has received growing attention due to demand from various emerging applications. In dynamic networks, one of the key interests is to model the event history of time-stamped interactions amongst nodes. We propose to model dynamic directed communication networks via multivariate counting processes. A pseudo partial likelihood approach is exploited to capture the network dependence structure.  Asymptotic results of the resulting estimation are established.  Numerical results are performed to demonstrate effectiveness of our proposal.  


\section{Introduction}\label{sect:intro}
Social networks have been actively studied in statistics literature in recent years. The methodological, theoretical and computational  developments thereof have been motivated by and led to interesting applications. Examples include understanding organizational structures from a human resources management perspective (\citealp[e.g.][]{Hopp_Zenk-2012-IJHRM}), transport planning based on travel behaviour (\citealp[e.g.][]{Kowald_etal-2010-Procedia, Cheng_etal-2015}), fraud detection (\citealp[e.g.][]{Baesens_etal-2015}), to name a few. Readers may refer to \cite{KarrerNewman2010}, \cite{Kolaczyk-2009}, \cite{YangEtal2012}, \cite{YenEtal2017} amongst others, for comprehensive overviews and recent developments on the statistical aspect of network data modelling across various disciplines. 

To understand the underlying dynamics of a social network of interest, one may make use of the corresponding event history data which include  interactions amongst participants. Survival analysis provides natural and effective tools for analyzing such data. Counting process techniques are typically applied for handling both time-to-event as well as recurrent event observations. \cite{AndersenGill1982} extended the  \cite{Cox1972} model for recurrent event time data and established the large sample properties of the corresponding estimators. As an alternative, \cite{PepeCai1993} and \cite{LawlessEtal1997} proposed the use of the mean function specification. In order to model multiple event times, \cite{Wei_etal-1989-JASA} developed a marginal approach,  while \cite{LinEtal2000} established large sample theory based on empirical processes. Readers may also refer to \cite{Andersen_etal-1993} and \cite{Martinussen_Scheike-2006} for summaries of martingale-based approaches and treatments for time-varying covariates.

We shall be concerned with a social network in which pairwise communication activities between senders and recipients are recorded by counting processes.  By modelling communication history as recurrent event time data, our proposal can entertain a more flexible dependence structure  amongst sequential events for each communication pair. In addition, it also attempts to incorporate another level of dependence amongst sender-recipient pairs.  We borrow the idea of composite likelihood for capturing the associated but unspecified dependence structure. To be more specific, our use of pseudo partial  likelihood is justified by the fact that distances amongst actors in social networks may not be as straightforwardly quantified as in traditional temporal and/or spatial data where temporal and/or geographical distances are naturally defined.  We establish a network version of $m$-dependent central limit theorem using Stein's \citep{Stein1972, Stein1986} method.  A refined asymptotic tightness of stochastic processes result is derived by allowing $m$ to grow with the sample size at a suitable rate. Convergence results of our proposed estimators are presented in Section \ref{sect:theory}.
 
The rest of the paper is organized as follows: Section~\ref{sec-methodology} presents the proposed method, discusses the inference procedure, and states the corresponding asymptotic results.  Extensive numerical results are presented in Section~\ref{sec-experiments}.  We conclude this paper with discussions in Section~\ref{sec-discussion}. Appendix provides details of the technical proofs of the results introduced in the main text.

\section{Methodology}\label{sec-methodology}

\subsection{Notation and Model}\label{sec-model}

Let $\mathcal{S} = \{1, \ldots, n\}$ denote the set of actors in a network. For each communication pair $(i, j) \in \mathcal{S}^{\otimes 2}$, we have its corresponding covariates $Z_{ij}(t) \in \mathbb{R}^p$, $t \in \mathcal{T}$, its corresponding event times $T_{ij, 0}, \ldots, T_{ij, n_{ij}}$, where $T_{ij, 0} = 0$ and each $n_{ij}$ is a non-negative integer-valued random variable.  We also assume that there is no self communication, i.e. for the pairs $(i, i)$, $i \in \mathcal{S}$, $n_{ii} = 0$.  For notational simplicity, we use $\mathcal{S}^{\otimes 2} = \{1, \ldots, n\}^{\otimes 2} \setminus \{(i, i), \, i \in \mathcal{S}\}$.  In the sequel, for a vector $v$, we denote $v^{\otimes 0} = 1$, $v^{\otimes 1} = v$ and $v^{\otimes 2} = vv^{\T}$. Let $\|v\|_1$ and $\|v\|$ be the $\ell_1$- and $\ell_2$-norms of $v$, respectively. For a $p \times q$ matrix $A=(a_{ij})$, we adopt the notation of $\|A\|_{\infty} = \max_{i=1, \ldots, p; j =1, \ldots, q}|a_{ij}|$ and $\|A\|_1 = \max_{i=1, \ldots, p}\sum_{j=1}^q|a_{ij}|$. For any set $\mathcal{B}$, we denote its cardinality by $|\mathcal{B}|$.

 For a subject $i \in \mathcal{S}$, we use multivariate counting processes to record its communication activities. Specifically, we adopt the notation
	\[
		\mathcal{N}_i = \{N_i(t), t\in \mathcal{T}\} = \{(N_{ij}(t), \, j \neq i), \, t\in \mathcal{T}\},
	\]
	where the univariate counting process $\{N_{ij}(t), t\in \mathcal{T}\}$ encodes the communication activities from sender $i$ to recipient $j$.  The corresponding mean function is given by
	\begin{equation}\label{eq-model-3}
		E\bigl\{dN_{ij}(t) \mid \bZ_{ij}(t)\bigr\} = \exp\bigl\{\bbeta^{o \T}\bZ_{ij}(t)\bigr\}\lambda_{0}(t)\,dt,
	\end{equation}
	where the baseline rate function $\lambda_{0}(\cdot)$ is unknown, and the coefficients vector $\beta^o \in \mathbb{R}^p$ is the parameters of interest.  We denote $\mathcal{F}_t$ as the $\sigma$-field generated by $\{N_{i}(s),\, i \in \mathcal{S}, \, 0 \leq s \leq t\}$.   In this paper we assume $p$ to be fixed, and more discussions regarding the dimensionality of $p$ can be found in Section~\ref{sec-discussion}.
 	
Different from \cite{PerryWolfe2013}, which assumes that each action in a network follows the Cox model and is conditionally independent with other previous events given past history, our formulation adopts the idea studied in \cite{LinEtal2000} which does not require accurate specifications of the dependence of sequential events within each pair. There is a subtle difference between \eqref{eq-model-3} and the celebrated Cox proportional hazards model which essentially assumes
	\begin{equation}\label{eq-model-cox}
		E\{dN_{ij}(t) \mid \mathcal{F}_{t-}\} = E\{dN_{ij}(t) \mid \bZ_{ij}(t)\}
	\end{equation} 
	in addition to \eqref{eq-model-3}. The requirement specified in \eqref{eq-model-cox} implies that the covariates included can capture all the dependence between the future and past events.  This is a valid assumption when at most one event occurs; however, for general information communication processes, where multiple events may happen over a certain period of interest, it is challenging to capture all the dependency by a set of well-conditioned covariates. Our formulation does not require \eqref{eq-model-cox} in which case the associated baseline function $\lambda_0(t)$ can be more generally defined under which random-effect intensity model can also be covered. It also considers a robust inference on the regression coefficients regarding the static covariates.  If we define $\{{M}_{ij}(\bbeta, t), \, t\in \mathcal{T}\}$ as 
	\begin{equation}\label{eq-dm-decomp}
		M_{ij}(\beta, t) = N_{ij}(t) - \int_0^t \exp\{\bbeta^{\T}\bZ_{ij}(s)\}\lambda_0(s)\,ds = N_{ij}(t) - \Lambda_{ij}(\beta, t),  
	\end{equation}
	then, due to \eqref{eq-model-3}, each $\{{M}_{ij}(\beta^o, t), \, t\in \mathcal{T}\}$ is a mean zero process instead of a martingale because \eqref{eq-model-cox} is no longer assumed. 	
	
To establish our inference procedure and introduce our proposed estimator, we mimic the idea of composite likelihood \citep[see, e.g.][]{Lindsay1988, CoxReid2004, VarinVidoni} and consider a pairwise pseudo partial likelihood. Since the observations are possibly dependent, this formulation is particularly useful when the full likelihood is too complicated to be expressed or optimized.  The corresponding log pseudo partial likelihood can be defined as follows,
	\begin{equation*}
		{\ell}_n(\bbeta) = \sum_{i=1}^n \sum_{j\neq i} \int_0^T \left\{\bbeta^{\T}\bZ_{ij}(t) - \log\left[\sum_{k = 1}^n \sum_{l \neq k} \exp\bigl\{\bbeta^{\T}\bZ_{kl}(t)\bigr\}\right]\right\} \,dN_{ij}(t),
	\end{equation*}
	whose score function, which is concave in $\beta$, is given by
	\[
		{U}_n(\beta) = \sum_{i=1}^n \sum_{j \neq i}{U}_{ij}(\beta) = \sum_{i = 1}^n\sum_{j \neq i} \int_0^T \left\{\bZ_{ij}(t) - \bar{\bZ}_n(\beta, t)\right\}\,dN_{ij}(t),
	\]
	where
	\[
		\bar{\bZ}_n(\beta, t) = \frac{\sum_{k = 1}^n \sum_{l \neq k} \bZ_{kl}(t)\exp\bigl\{\bbeta^{\T}\bZ_{kl}(t)\bigr\}}{\sum_{k = 1}^n \sum_{l \neq k} \exp\bigl\{\bbeta^{\T}\bZ_{kl}(t)\bigr\}}.
	\]

As shown in \ref{sect:theory}, the score function $U_n(\beta^0)$, if suitably normalized, is asymptotically normal with mean zero. This makes ${U}_n(\bbeta) = 0$  a valid consistent estimation equation.  In the sequel, we define $\hat{\bbeta}$ as the solution that solves 
		\begin{equation}\label{eq-beta-hat}
				U_n(\hat{\bbeta}) = 0.			
		\end{equation}

\subsection{Theory}\label{sect:theory}

We first list the regularity assumptions. 
	
\begin{condition}\label{assumption-model}
	For any $i \in \mathcal{S}$,  there exists $\mathcal{J}_i \subset \mathcal{S}$, such that for any $j \in \mathcal{S} \setminus \mathcal{J}_i$, $\mathcal{N}_i$ and $\mathcal{N}_j$ are independent, and $\mathcal{Z}_i$ and $\mathcal{Z}_j$ are independent, where $\mathcal{Z}_{i} = \{Z_{il}(t), l \in \mathcal{S}, t \in \mathcal{T}\}$.  Assume for any $i \in \mathcal{S}$, it holds that
	\begin{equation}\label{eq-con1}
		|\mathcal{J}_{i}| \asymp m_n = o(n^{1/4}).
	\end{equation}
\end{condition}

\begin{condition}\label{assumption-regular}
	Assume for all $(i, j)\in \mathcal{S}^{\otimes 2}$, there exists an absolute constant $K > 0$ such that
		\[
			\|Z_{ij}(0)\|_1 + \int_0^T\,\|dZ_{ij}(t)\|_1 \leq K.
		\]
		Let 
		\[
			\mu_n(\beta^o, t) = \frac{E \bigl\{\sum_{i=1}^n\sum_{j\neq i}Z_{ij}(t)\exp(\beta^{o\T}Z_{ij}(t))\bigr\}}{E\bigl\{\sum_{i=1}^n\sum_{j\neq i}\exp(\beta^{o\T}Z_{ij}(t))\bigr\}},
		\]
		\[
			\Sigma_{1,n} = E\left[\frac{\sum_{i=1}^n\sum_{j\neq i}}{n(n-1)}\int_0^T \{Z_{ij}(t) - \mu_n(\beta^o, t)\}^{\otimes 2} \exp(\beta^{o\T}Z_{ij}(t))\lambda_0(t)\,dt\right],
		\]
		and
		\[
			\Sigma_{2, n}  = E\Bigg(\frac{1}{n(n-1)}\left[\sum_{i=1}^n\sum_{j\neq i}\int_0^T \left\{Z_{ij}(t) - \mu_{n}(\beta^o, t)\right\}\, dM_{ij}(\beta^o, t)\right]^{\otimes 2}	\Bigg), 
		\]
		satisfying that 
		\[
			0 < \limsup_{n\rightarrow\infty}\frac{\rho_{\text{max}}(\Sigma_{j,n})}{\rho_{\text{min}}(\Sigma_{j,n})} < \infty, \quad j = 1, 2,
		\]
		where $\rho_{\min}(A)$ and $\rho_{\max}(A)$ are the minimum and maximum eigenvalues of matrix $A$, respectively.
		
	In addition, we assume that there exists a non-random vector $\mu(t)$ and matrices $\Sigma_1$, $\Sigma_2$ such that 
			\[
			\sup_{t\in \mathcal{T}}\max\left\{\|\mu_n(\beta^o, t) - \mu(t)\|, \, \bigl\|\Sigma_{1, n} - \Sigma_1\bigr\|, \, \bigl\|\Sigma_{2, n} - \Sigma_2\bigr\| \right\}\stackrel{P}{\to} 0.
			\]
\end{condition}

Condition~\ref{assumption-model} is the key assumption in this paper and is reasonably mild.  It restricts the growing rate of $m_n$ as an $m$-dependent assumption.  We adopt this modified $m$-dependent assumption in our network formulation because there is no linear ordering nor a natural distance as one has in the time series context.  To ensure that this assumption is realistic in our applications concerned, we also allow the dependence number $m_n$ to diverge with $n$.   If we require all edges are independent, which is commonly assumed in the stochastic block model, then we have $m_n = 1$.  In fact, in order to show the asymptotic normality of the estimators in Theorem~\ref{thm1}, we need to show a finite-dimensional central limit theorem, which requires $\max_{i\in \mathcal{S}}|\mathcal{J}_i| \leq m_n = o(n^{1/4})$, and the tightness of relevant processes, which requires for any $i\in \mathcal{S}$, $|\mathcal{J}_i| \asymp m_n \leq O(n^{1/3})$.   

Condition~\ref{assumption-regular} requires specific covariance structure that guarantees its positive definiteness. We define $\Sigma_{1, n}$ and $\Sigma_{2, n}$ in this way such that we can apply the self-normalizing version of the central limit theorem, which works under mild assumptions on the dependence structure.  Condition~\ref{assumption-regular} also implies that the mean processes $\{\Lambda_{ij}(\beta^o, t), \, t\in\mathcal{T}\}$ are Lipschitz continuous.  This fact will be repeatedly used in the proof in the Appendix.

\begin{theorem}\label{thm1}
	Under Conditions~\ref{assumption-model} and \ref{assumption-regular}, $\hat{\beta}$ defined in \eqref{eq-beta-hat} and $\Sigma_{1, n}, \Sigma_{2, n}$ specified in Condition~\ref{assumption-regular}, we have $\hat{\beta}$ asymptotically normal with mean $\beta^o$ and covariance matrix $\Sigma_n = \Sigma^{-1}_{1, n}\Sigma_{2,n}\Sigma^{-1}_{1, n}$. In particular, for any $v \in \mathbb{R}^p$, 
	\begin{equation}\label{eq-thm1-result}
		\frac{{v}^{\T}(\hat{ {\beta}} -  {\beta}^o)}{\left(v^{\T} \Sigma_n^{-1}v\right)^{1/2}} \to \mathcal{N}(0, 1), 
	\end{equation}
	in distribution, as $n\to \infty$.  
\end{theorem}

Theorem \ref{thm1} states that the convergence rate is related with $m_n$. It follows from \eqref{eq-proof-1} in the proof that the convergence rate of the variance $v^{\T}\Sigma_n^{-1}v, v \in \mathbb{R}^p$,  which is $O(m_n^{2}n^{-1/2})$.  If $m_n = 1$ as assumed in the independent edges cases or $m_n = O(1)$ as in \cite{SchweinbergerHandcock2015}, then the convergence rate is $n^{-1/2}$, which is the same as the standard situation and is optimal.  The convergence rate decreases as $m_n$ increases.  We would like to point out that \cite{SchweinbergerHandcock2015} also considered dependent network models, and the results developed therein are based on a general exponential random graph model in a Bayesian framework.  The key differences between our paper and theirs are summarized as follows: i) the random variable associated with each edge is a Bernoulli random variable representing the presence of the edge in \cite{SchweinbergerHandcock2015}, while in our paper, each edge has its own counting process; ii) although both papers allow $m$-dependence, in \cite{SchweinbergerHandcock2015} $m$ is assumed to be finite, while in our setting, $m$ is allowed to be of order $o(n^{1/4})$; and iii) since we are under a survival analysis framework, the dependence structure is assumed to be among the senders, while \cite{SchweinbergerHandcock2015} investigated exponential random graph models, and therefore the dependence lies among edges.  

It is also natural to estimate the mean function $\Lambda_{0}(t)$ for $i \in \mathcal{S}$ by the Aalen--Breslow-type estimator
    \begin{equation*}
		\hat{\Lambda}_{0}(t) = \int^t_0\frac{d\sum^n_{i=1}\sum_{i\neq j}N_{ij}(s)}{\sum^n_{i=1}\sum_{i\neq j}\exp\{\hat{\beta}^\T Z_{ij}(s)\}}, \quad t \in [0, T],
    \end{equation*}
	which can be shown to be consistently estimated, by proving $\hat{\beta}$ is almost surely consistent in {\bf Step 1} in the proof. This requires a strong law of large numbers for dependent random variables \citep[see e.g.][]{KorchevskyPetrov2010}. The result  can also be generalize further so that individual senders can have different baseline hazards.

The result \eqref{eq-thm1-result} has the same sandwich variance form which resembles those obtained via the composite likelihood inference \citep[e.g.][]{LinEtal2000}. It degenerates to the efficient estimator case when $\Sigma_{1, n} = \Sigma_{2, n}$, i.e. all the directed edges are assumed to be independent.

There are a few key ingredients in the proof of Theorem~\ref{thm1}.  Firstly, we exploit the Stein's method, which has been extended to dependent cases in \cite{BaldiRinott1989}, \cite{ChenEtal2010}, to name but a few.  We leave the proof in the Appendix.  It worths to point out that we have developed a new device of central limit theorem which is designed for the semiparametric setting, integrating chaining arguments.  Secondly, we extend the weak convergence proof in the independent case stated in \cite{LinEtal2000} to a dependent case.  Since we allow $m_n$ to diverge, this is not a straightforward extension.

\subsection{Variance estimators}

We have established the asymptotic of $\hat{\bbeta}$ in Theorem~\ref{thm1}, but it involves unknown population quantities $\Sigma_{1, n}$ and $\Sigma_{2, n}$.  Usual estimators are based on the assumption of independent observations or independent innovations in the time series context. To tackle the unknown dependence structure, we adopt the jackknife sandwich estimator proposed in the composite likelihood literature; see e.g. \cite{VarinEtal2011} in an attempt to provide reasonable estimation for $\Sigma_{1, n}$ and $\Sigma_{2, n}$.

Let $\widehat{\Sigma}_{1, n}$ and $\widehat{\Sigma}_{2, n}$ be the estimators of $\Sigma_{1, n}$ and $\Sigma_{2, n}$, respectively.  They are defined as follows:
	\[
	\widehat{\Sigma}_{1, n} =  \sum_{i=1}^n \sum_{j\neq i} \sum_{k = 1}^{n_{ij}} \left\{Z_{ij}(T_{ij, k}) - \bar{Z}\bigl(\hat{\bbeta}, T_{ij, k}\bigr)\right\}^{\otimes 2} \delta_{ij},
	\]
	and
	\begin{equation}\label{eq-est-variability}
	\widehat{\Sigma}_{2, n} = \frac{1}{n}\sum_{s=1}^n \left[\sum_{i \neq s} \sum_{j \neq i, s}\sum_{k=1}^{n_{ij}} \bigl\{Z_{ij}(T_{ij, k}) - \bar{Z}(\hat{\beta}^{(-s)}, T_{ij, k})\bigr\}^{\delta_{ij}} \right]^{\otimes 2},
	\end{equation}
	where $\delta_{ij} = \mathbbm{1}\{N_{ij}(T) > 0\}$, $\hat{\bbeta}^{(-s)}$ is the estimator to the estimating equation \eqref{eq-beta-hat} after deleting the $s$th node and its corresponding data from the observations. The proposed jackknife procedure offers one possible approach for variance estimation in Theorem \ref{thm1}. Although there is no single construction of consistent estimates for $\Sigma_{2,n}$ and hence the variance, unless there is additional network and time dependence structural assumption, as we shall demonstrate in Section \ref{sec-experiments}, the jackknife procedure performs decently in terms of achieving empirical coverage probabilities that are very close to the nominal values.

In fact, for a general composite likelihood inference problem, estimation of the variability matrix $J_o$ is often challenging. The essence of composite likelihood is to make use of the working independence and sandwich variance estimator to capture the dependency so that specific dependent structures need not be assumed; however, in practice, a good estimator of the variance is inevitably a function of the unknown structure.  One feasible quick remedy of this dilemma is the jackknife estimator of the variability matrix as defined in \eqref{eq-est-variability}.

Consequently, for $\alpha \in (0, 1)$, a $(1-\alpha)\times 100\%$ percent confidence region for the true regression parameter $\bbeta^o$ can, be expressed as 
	\[
	\bigl\{\beta: (\hat{\beta}-\beta^o)^\T\left(\widehat{\Sigma}_{1, n}\Sigma_{2,n}^{-1}\widehat{\Sigma}_{1, n}\right)^{-1}(\hat{\beta}-\beta^o)\leq \chi^2_{p,1-\alpha}\bigl\},
	\] 
	where $\chi^2_{p, 1-\alpha}$ is the corresponding quantile of a chi-square distribution with degrees of freedom $p$.

\section{Numerical analysis}\label{sec-experiments}
\subsection{Simulation}
A  simulation study is presented in this section. We generated recurrent events for each pair of nodes of size 150 in a connected network from the model
\begin{equation*}
   \lambda_{ij}(t; Z) = \eta_i\lambda_{0i}(t)\exp\{ {\bbeta}_0^\T Z_{ij}\},	
\end{equation*}
where $\lambda_{0i}(t) = \mathbbm{1}(i\leq n/2) + 1.2\mathbbm{1}(i>n/2)$ denotes the baseline hazard. The unobserved random variable $\eta$ that introduces heterogeneity to this random-effect intensity model was assumed to follow Gamma distribution with mean and variance equal 1 and 1/16, respectively. Denote $N = n(n-1)$. We considered two sets of time-invariant covariates. The first set was generated $\widetilde{ {Z}}^\T = ( {Z}_1^\T,  {Z}_2^\T,  {Z}_3^\T) = (Z_{ij})_{N\times 3}$ in the following way: $  Z_1^\T\T = ( {Z}_{11\cdot}^\T , \ldots,  {Z}_{1n\cdot}^\T )^\T  = (\mathbbm{1}( \Xi_1 \geq   0)^\T , \ldots, \mathbbm{1}( \Xi_n \geq   0)^\T )^\T$, where $\geq$ denotes the element-wise comparison, $ \Xi_j$ is an independent and identically distributed $(n-1)$-dimensional Gaussian random variable with mean $0$ and $ \Sigma^{(1)} = (\sigma^{(1)}_{ij}) = \mathbbm{1}(i=j) + \rho \mathbbm{1}(i\neq j)$ for $j = 1, \ldots, n$.  Covariate $ {Z}_2$ were generated from $\textsc{Unif}(0,1)$ and $  Z_3^\T = (Z_{31\cdot}^\T , \ldots, Z_{3n\cdot}^\T )^\T $ were independent and identically distributed as normally distributed with mean $  0$ and variance $ \Sigma^{(2)}$ with $(\sigma^{(2)}_{ij}) = I(i=j) + \rho I(|i-j|=1)$, which is a band matrix with bandwidth 1. 

To examine the robustness of our approach, we also considered another set of simulations in which the covariates are dependent amongst senders. The setting is exactly the same as the first set discussed previously except that the covariates $\widetilde{{Z}}$ are correlated. In this study, $  Z_1^\T = (Z_{11}, \ldots, Z_{1N}) = I( \Xi \geq   0)$, where $ \Xi$ is a $N$-dimensional Gaussian random variable with mean $  0$ and $ \Sigma^{(1)} = (\sigma^{(1)}_{ij}) = I(i=j) + \rho I(i\neq j)$; $Z_2$ were generated as an $N$-vector of independent and identically distributed as $\textsc{Unif}(0,1)$ random variables and $  Z_3^\T = (Z_{31}, \ldots, Z_{3N})$ were normally distributed with mean $  0$ and variance $ \Sigma^{(2)}$ with $(\sigma^{(2)}_{ij}) = I(i=j) + \rho I(|i-j|=1)$, which is a band matrix with bandwidth 1.

For the jackknife procedure, both odd-1-out and odd-2-out procedures were examined: The odd-2-out procedure randomly remove two nodes from the network for estimating the corresponding variances based on 150 random draws. The results of the simulation studies are summarized in Table \ref{tbl:simind} based on 250 iterations.

The estimates provided by the proposed method are virtually unbiased while the variance estimator also provides reasonably accurate estimation of the true variances of $\hat{\beta}$ upon which the confidence intervals constructed demonstrate empirical coverage probabilities that are close to their nominal values.

\begin{table}[]
\centering
\caption{Summary statistics for the simulation studies$^\dagger$}
\label{tbl:simind}\footnotesize
\begin{tabular}{cccccccccc}\hline\hline
	$\rho$  & Parameters    & Bias   & SE    & SEE(JK) & SEE(JK2) & SEE   & ECP(JK) & ECP(JK2) & ECP   \\[5pt]
	\multicolumn{10}{c}{Independent} \\
	0$\cdot$00 & $\beta_{10}$ & 0$\cdot$015  & 0$\cdot$016 & 0$\cdot$016   & 0$\cdot$016    & 0$\cdot$016 & 0$\cdot$952    & 0$\cdot$952   & 0$\cdot$952  \\
	& $\beta_{20}$ & -0$\cdot$005 & 0$\cdot$022 & 0$\cdot$024   & 0$\cdot$024    & 0$\cdot$025 & 0$\cdot$952    & 0$\cdot$952   & 0$\cdot$960 \\
	& $\beta_{30}$ & 0$\cdot$002  & 0$\cdot$007 & 0$\cdot$007   & 0$\cdot$007    & 0$\cdot$007 & 0$\cdot$920    & 0$\cdot$920   & 0$\cdot$936 \\
	0$\cdot$30 & $\beta_{10}$  & 0$\cdot$013  & 0$\cdot$019 & 0$\cdot$019   & 0$\cdot$019    & 0$\cdot$019 & 0$\cdot$962   & 0$\cdot$962    & 0$\cdot$872  \\
	& $\beta_{20}$  & -0$\cdot$004 & 0$\cdot$026 & 0$\cdot$025   & 0$\cdot$025    & 0$\cdot$025 & 0$\cdot$936   & 0$\cdot$936    & 0$\cdot$956 \\
	& $\beta_{30}$  & 0$\cdot$001  & 0$\cdot$008 & 0$\cdot$007   & 0$\cdot$007    & 0$\cdot$007 & 0$\cdot$952   & 0$\cdot$952    & 0$\cdot$900 \\
	0$\cdot$50 & $\beta_{10}$  & 0$\cdot$018  & 0$\cdot$033 & 0$\cdot$046   & 0$\cdot$046    & 0$\cdot$042 & 0$\cdot$973   & 0$\cdot$973    & 0$\cdot$856  \\
	& $\beta_{20}$  & -0$\cdot$003 & 0$\cdot$028 & 0$\cdot$025   & 0$\cdot$025    & 0$\cdot$026 & 0$\cdot$928   & 0$\cdot$928    & 0$\cdot$944 \\
	& $\beta_{30}$  & 0$\cdot$002  & 0$\cdot$008 & 0$\cdot$007   & 0$\cdot$007    & 0$\cdot$007 & 0$\cdot$924   & 0$\cdot$928    & 0$\cdot$944 \\
	\multicolumn{10}{c}{Dependent} \\
	0$\cdot$30 & $\beta_{10}$ & 0$\cdot$015  & 0$\cdot$019 & 0$\cdot$019   & 0$\cdot$019    & 0$\cdot$019 & 0$\cdot$945    & 0$\cdot$945   & 0$\cdot$945  \\
	& $\beta_{20}$ & -0$\cdot$006 & 0$\cdot$023 & 0$\cdot$024   & 0$\cdot$024    & 0$\cdot$025 & 0$\cdot$958    & 0$\cdot$958   & 0$\cdot$972 \\
	& $\beta_{30}$ & 0$\cdot$002  & 0$\cdot$007 & 0$\cdot$007   & 0$\cdot$007    & 0$\cdot$007 & 0$\cdot$944    & 0$\cdot$944   & 0$\cdot$944 \\
	0$\cdot$50 & $\beta_{10}$ & 0$\cdot$014  & 0$\cdot$026 & 0$\cdot$035   & 0$\cdot$035    & 0$\cdot$033 & 0$\cdot$940    & 0$\cdot$940   & 0$\cdot$933  \\
	& $\beta_{20}$ & -0$\cdot$002 & 0$\cdot$025 & 0$\cdot$025   & 0$\cdot$025    & 0$\cdot$025 & 0$\cdot$953    & 0$\cdot$953   & 0$\cdot$960 \\
	& $\beta_{30}$ & 0$\cdot$001  & 0$\cdot$008 & 0$\cdot$007   & 0$\cdot$007    & 0$\cdot$007 & 0$\cdot$940    & 0$\cdot$940   & 0$\cdot$947 \\
	\hline\hline
\end{tabular}
\caption*{\footnotesize  $^\dagger$Bias is the mean differences between parameter estimates and their corresponding true values, SE denotes the standard errors of the parameter estimates; SEE(JK2), SEE(JK), SEE are the means of the variance estimates under odd-two-out, odd-one-out jackknife and naive standard error estimates, respectively; the corresponding empirical coverage probabilities are denoted as ECP(JK2), ECP(JK) and ECP, respectively. The na\"ive standard error estimates are based on $\widehat{\Sigma}_{1, n}^{-1}$. }
\end{table}

\subsection{A real data analysis example}
In this section, we applied our proposed procedure to examine a celebrated set of network data, namely the Enron e-mail corpus. This data set is one of the largest e-mail communication datasets that include not only both the senders and the recipients information, but also the times at which individual emails were sent. The dataset was originally compiled by the Federal Energy Regulatory Commission (FERC) and was released in 2002 after the bankruptcy filed by the company. The email corpus contains the information about 156 employees amongst whom 21,635 messages were delivered between 13\textsuperscript{th} November 1998 and 21\textsuperscript{st} June  2002. Individual attributes including their department affiliations (Legal, Trading, or other), seniorities (junior or senior) and genders are also recorded. 

Similar to the study investigated in \cite{PerryWolfe2013}, we excluded \textit{en masse} messages so that messages with more than five recipients were removed. These messages make up an approximately 30\% of all the messages. Upon these data points, we applied model \eqref{eq-model-3} to analyse the data with transformed timestamps $\log T$. In particular, in order to examine whether or not there exists homophily in the network, \textit{i.e.} whether or not common traits shared between a pair of individuals in the network (in our case, department, seniority and gender) are statistical significant with respect to the prediction of their interactive behaviour, we included three indicator variables with values 1 if a pair shares the same department, seniority and gender, respectively. The corresponding fits and variance estimates are summarized in Table \ref{tbl:data}. From the summary, we can observe that the covariates that capture the effect of same department and seniority between two members in the network are statistically significant. Our results agree with those findings of \cite{PerryWolfe2013}.  In addition, we also observe that the standard errors estimated via jackknife methods are smaller than those obtained without using the jackknife correction. This can be due to the fact that a certain level of negative correlation between covariates is captured by the jackknife approach. Although in this particular example, the reduced standard error estimates do not affect our conclusion in terms of covariates' statistical significance, we emphasize here that for general applications, it could be rather challenging to fully capture the actual dependence of the network of interest via a selected set of covariates. The jackknife adjustment upon the sandwich estimator here provides a computationally feasible alternative for problems of this sort. Finally, it is worth mentioning that the computation time for a dataset of this scale using a standard desktop computer with a 3$\cdot$40-GHz processor and 16Gb of random access memory is merely around 100 seconds. 

\begin{table}[]
\centering
\caption{Summary statistics for the Enron email dataset$^\ddagger$}
\label{tbl:data}\footnotesize
\begin{tabular}{ccccc}\hline\hline
Parameters       & Estimate & SEE(JK) & SEE(JK2) & SEE \\
$\beta_{Dept}$   &  0$\cdot$858  & 0$\cdot$235  & 0$\cdot$235 & 0$\cdot$163\\
$\beta_{Junior}$ &  0$\cdot$254  & 0$\cdot$107  & 0$\cdot$106 & 0$\cdot$159\\
$\beta_{Gender}$ &  0$\cdot$171  & 0$\cdot$110  & 0$\cdot$110 & 0$\cdot$156\\\hline\hline\end{tabular}

\caption*{\footnotesize $^\ddagger$Estimate corresponds to the point estimate for the parameters; SEE(JK2), SEE(JK), SEE are the variance estimates under odd-two-out, odd-one-out jackknife and naive standard error estimates, respectively. The na\"ive standard error estimates are based on $\widehat{\Sigma}_{1,n}^{-1}$. }
\end{table}

\section{Discussion}\label{sec-discussion}
In conclusion, this paper studies dynamic and directed communication networks under a mild assumption on the dependence structure amongst edges in the network.  More specifically, upon the recurrent event survival analysis structure, we establish a network version $m$-dependent central limit theorem in which $m$ grows with the sample size $n$. The robust sandwich variance estimator as well as our jack-knife procedure are justified numerically.   This general framework and its asymptotic results are new and unique.   

We would like to highlight that in this paper, we assume the dimension of the regression coefficient to be fixed, i.e. $\beta^o \in \mathbb{R}^p$, and $p$ is fixed.  One could incorporate high-dimensional inference techniques developed \citep[e.g.][]{HuangEtal2013}.  It is beyond the scope of this paper and we would like to make the setting simple in order to make the idea clearer.

Indeed, if for some real data sets, where a reasonable distance is available for the hyper graph, one can also use the results in random fields to establish the limiting distribution of the estimators with efficient variance estimators.  We would leave this as future work when proper data sets with definitions of distance properly defined.

\section*{Appendix}

In the Appendix, we will first show the tightness based on our $m$-dependence assumption, and then proof the main results in Theorem~\ref{thm1}.

\begin{lemma}[Tightness]\label{lem-1}
	For $t \in (0, T]$, define
	\[
	\bar{M}(t) = \frac{\sum_{i=1}^n \sum_{j\neq i} M_{ij}(\beta^o, t)}{\sqrt{\mathrm{var}\bigl\{\sum_{i=1}^n \sum_{j\neq i} M_{ij}(\beta^o, T)\bigr\}}}.
	\]
	Under Conditions~\ref{assumption-model} and \ref{assumption-regular}, we have that for any $\epsilon, \eta > 0$, there exists $\delta > 0$, such that
	\[
	\limsup_{n\to \infty} \mathrm{pr}\left\{\sup_{|s-t| < \delta}\left|\bar{M}(t) - \bar{M}(s)\right| \geq \eta \right\} < \epsilon.
	\]
\end{lemma}

\begin{proof}

Without loss of generality, let $0 < s \leq t \leq T$.  Define by convention that $0/0 = 0$.  For notational simplicity, in this proof, for any $i = 1, \ldots, n$, let $M_i(t) = \sum_{j \neq i} M_{ij}(\beta^o, t)$.
	
It follows from the Cauchy--Schwarz inequality that 
	\begin{align*}
		\sigma(t) & = \sqrt{\mathrm{var}\left\{\sum_{i=1}^n M_{i}(t)\right\}} \leq \left(E\left[\left\{\sum_{i=1}^n M_{i}(t)\right\}^4\right]\right)^{1/4} \\
		& = \left(\sum_{i, j, k, l = 1}^n E\bigl(M_i(t)M_j(t)M_k(t)M_l(t)\bigr) \right)^{1/4}.
	\end{align*}
	
Among all possible $(i, j, k, l) \in \{1, \ldots, n\}^{\otimes 4}$, it follows from Condition~\ref{assumption-model} and the proof of Corollary~2 in \cite{BaldiRinott1989} that there are at most $O(nm^3_n)$ terms consisting of 4 distinct elements having non-zero means, at most $O(nm^2_n)$ terms consisting of 3 distinct elements having non-zero means, $O(n^2)$ terms consisting of 2 distinct elements having non-zero means and $O(n)$ terms consisting of 1 distinct element having non-zero means.  Therefore, due to Conditions~\ref{assumption-model} and \ref{assumption-regular}, it holds that, for $t \in (0, T]$,
	\begin{equation}\label{eq-sigma-lemma}
		 \sigma(t) = O \bigl\{\bigl(nm_n^3 + n^2\bigr)^{1/4}\bigr\}.
	\end{equation}

Since the denominator of $\bar{M}(t)$ is also a function of $t$, we first decompose the difference as follows,
	\begin{align}\label{eq-barM-st}
		\bigl|\bar{M}(t) - \bar{M}(s)\bigr| \leq  \left|\frac{\sum_{i=1}^n \bigl(M_{i}(t) - M_{i}(s)\bigr)}{\sigma(T)}\right|.
	\end{align}

It follows from \eqref{eq-sigma-lemma} that up to an absolute constant, it suffices to show that
	\[
		\mathrm{pr}\left\{\sup_{0 < t - s < \delta}\left|W_n(t) - W_n(s)\right| > \eta/2 \right\} < \epsilon,
	\]
	where 
	\[
		W_n(t) = \frac{1}{\bigl(nm_n^3 + n^2\bigr)^{1/4}}\sum_{i=1}^n M_i(t).
	\]
	For any $k \in \{0, 1, 2, \ldots, \}$, define sets 
	\[
	A_k(\delta) = \{i\delta/2^k,\, i = 0, 1, 2, \lfloor 2^kT/\delta\rfloor\} \cup \{T\},
	\]
	and $t_k = \min \bigl\{u \in A_k(\delta), \, u \geq t\bigr\}$, $s_k = \max \bigl\{u \in A_k(\delta), \, u \leq s\bigr\}$.

For a large enough $K \geq 2$ satisfying $2^K \geq n\bigl(nm_n^3 + n^2\bigr)^{-1/4}$ and $\eta' = \eta/2$, we have
	\begin{align*}
	& \mathrm{pr}\left\{\sup_{t-s < \delta} \bigl|W_n(s) - W_n(t)\bigr| \geq \eta'\right\} \\	
	\leq & \mathrm{pr}\left\{\sup_{t-s < \delta} \bigl(\bigl|W_n(s) - W_n(s_K)\bigr| + \bigl|W_n(s_K) - W_n(t_K)\bigr| + \bigl|W_n(t) - W_n(t_K)\bigr|\bigr) \geq \eta'\right\} \\
	\leq & \mathrm{pr}\left\{\sup_{0 \leq t_K - t \leq \delta/2^K}\max_{t_K \in A_K(\delta)} \bigl|W_n(t) - W_n(t_K)\bigr| \geq \eta'/4\right\}  \\
	 + & \mathrm{pr}\left\{\sup_{0 \leq s - s_K \leq \delta/2^K}\max_{s_K \in A_K(\delta)} \bigl|W_n(s) - W_n(s_K)\bigr| \geq \eta'/4\right\} \\
	 + & \mathrm{pr}\left\{\max_{s_K \in A_K(\delta), 0 \leq  t_K - s_K \leq 2\delta }|W_n(s_K) - W_n(t_K)| \geq \eta'/2\right\} = (I) + (II) + (III).
	\end{align*}

Note that (I) and (II) can be dealt with using the same arguments, therefore we will only deal with terms (I) and (III).
	
For (I), define $t'_K = \max\bigl\{u \in A_K(\delta), \, u < t_K\bigr\}$.  We have $t'_K \leq t \leq t_K$ and $t_K - t'_K = \delta/2^K$.  Therefore, by the monotonicity of $N(\cdot)$ and $\Lambda(\cdot)$, one can write
	\begin{align*}
	&	W_n(t_K) - W_n(t) \leq \bigl(nm_n^3 + n^2\bigr)^{-1/4}\sum_{i=1}^n\bigl\{N_i(t_{K}) - N_i(t'_{K}) + \Lambda_i(t_{K}) - \Lambda_i(t'_{K})\bigr\} \\
	& = \bigl(nm_n^3 + n^2\bigr)^{-1/4}\sum_{i=1}^n\bigl\{\bigl(N_i(t_{K}) - \Lambda_i(t_{K})\bigr) - \bigl(N_i(t'_{K}) - \Lambda_i(t'_{K})\bigr) + 2\bigl(\Lambda_i(t_{K}) - \Lambda_i(t'_{K})\bigr)\bigr\} \\
	& = W_n(t_K) - W_n(t'_{K}) + 2\bigl(nm_n^3 + n^2\bigr)^{-1/4}\sum_{i=1}^n\bigl(\Lambda_i(t'_{K}) - \Lambda_i(t_{K})\bigr).
	\end{align*}

It follows from Assumption~\ref{assumption-regular} that there exists a small enough $\delta_2 > 0$ such that for any $0 < \delta < \delta_2$
	\[
	2\max_{t_K \in A_K(\delta)}\bigl(nm_n^3 + n^2\bigr)^{-1/4}\sum_{i=1}^n\bigl(\Lambda_i(t'_{K}) - \Lambda_i(t_{K})\bigr) \leq 2Cn\bigl(nm_n^3 + n^2\bigr)^{-1/4}\delta/2^K \leq \eta/8.
	\]
	Then,
	\begin{align*}
	(I) \leq \frac{T2^K}{\delta} \frac{256}{\eta^4} E\left\{|W_n(t_K) - W_n(t'_K)|^4\right\} \leq \frac{2^K C^2\times 256 \bigl(nm^3 + n^2\bigr) \delta^2}{\delta \eta^4 \bigl(nm^3 + n^2\bigr) \times 4^K}  = \frac{256C^2\delta}{2^K\eta^4} \leq \epsilon/4,
	\end{align*}
	where the second inequality follows the same arguments leading to \eqref{eq-sigma-lemma}.

It suffices to show $(III) \leq \epsilon/2$.  Due to our construction of $A_k(\delta)$ for $k = 0, 1, 2, \ldots$, we have $|t_{k+1} - t_k| \leq 2^{-(k+1)} \delta$, therefore for $K$ used above, it holds that
	\[
	|t_K - t_0| \leq \delta\sum_{k = 0}^{K} 2^{-k} \leq 2\delta.
	\]
	Same arguments lead to $|s_K - s_0| \leq 2\delta$.  Since we chose $K$ in the way that $|s_K - t_K| \leq 2\delta$, it holds that $|s_0 - t_0| \leq 6\delta$.  Then,
	\begin{align*}
	|W_n(s_K) - W_n(t_K)| & \leq \sum_{k = 1}^K |W_n(t_k) - W_n(t_{k-1})| + \sum_{k = 1}^K |W_n(s_k) - W_n(s_{k-1})| \\
	& + |W_n(t_0) - W_n(s_0)|  = (III.1) + (III.2) + (III.3).
	\end{align*}
	It suffices to bound the probabilities of (III.1) and (III.3) being large.
	
As for (III.3), note that
\[
\mathrm{pr}\left\{\max_{\stackrel{t_0, s_0 \in A_0(\delta)}{|t_0 - s_0| \leq 6\delta}}|W_n(t_0) - W_n(s_0)| \geq \eta/4\right\} \leq \frac{C^2\bigl(nm_n^3 + n^2\bigr) \times 256\delta^2}{\delta \bigl(nm_n^3 + n^2\bigr)\eta^4} = \frac{256C^2\delta}{\eta^4} \leq \epsilon/4.
\]	

As for (III.1), note that
\begin{align*}
& \mathrm{pr}\left\{ \sum_{k=1}^K\max_{t_k \in A_k(\delta)} |W_n(t_k) - W_n(t_{k-1})| \geq \eta/8 \right\}	\leq \frac{8}{\eta}E\left|\sum_{k=1}^K\max_{t_k \in A_k(\delta)}  |W_n(t_k) - W_n(t_{k-1})|\right| \\
\leq &   \frac{8}{\eta}\sum_{k=1}^K \left\{|A_k(\delta)|\max_{t_k \in A_k(\delta)}  E\left( |W_n(t_k) - W_n(t_{k-1})|^4\right)\right\}^{1/4} \\
\leq & \frac{8}{\eta} \sum_{k=1}^K \left\{\frac{2^k}{\delta} \frac{C^2\delta^2}{4^k}\right\}^{1/4} \leq \frac{8C^{1/2}}{\eta}\delta^{1/2} \leq \epsilon/4.
\end{align*}

The final conclusion holds by combining all the terms above.	

\end{proof}

\begin{proof}[of Theorem~\ref{thm1}]
This proof consists of three steps.  We first show the consistency of $\hat{\beta}$, then the asymptotic normality of $U_n(\beta^o)$, and lastly to control the residuals involved.  To begin, we define some additional notation.  Let 
	\[
		S^{(k)}(\beta, t) = \sum_{i = 1}^n \sum_{j\neq i} \big\{Z_{ij}^{\otimes k}(t) w_{ij}(\beta, t)\big\},
	\]
	where $k = 0, 1, 2$ and $w_{ij}(\beta, t) = \exp(\beta^{\T}Z_{ij}(t))$.  

\vskip 3mm
\noindent \textbf{Step 1: Consistency.}  In order to prove the consistency of $\hat{\beta}$, we recall that $\hat{\beta}$ is the root of equation ${U}_n\bigl(\hat{\beta}\bigr) = 0$.  Define
	\begin{align*}
	X(\beta) & = \frac{1}{n(n-1)}\bigl(\ell_n(\beta) - \ell_n(\beta^o)\bigr) \\
	& = \frac{1}{n(n-1)}\sum_{i=1}^n \sum_{j \neq i} \int_0^T \left\{(\beta - \beta^o)^{\T}Z_{ij}(t) - \log \left[\frac{S^{(0)}(\beta, t)}{S^{(0)}(\beta^o, t)}\right]\right\}\, dN_{ij}(t) \\
	& = \frac{1}{n(n-1)}\sum_{i=1}^n \sum_{j \neq i} \int_0^T X_{ij}(\beta, t) \, dN_{ij}(t).
	\end{align*}
	With probability tending to 1, its unique maximum is attained at $\hat{\beta}$ due to the definition of $\hat{\beta}$.  Moreover, let 
	\begin{align*}
	A(\beta) = \frac{1}{n(n-1)}\sum_{i=1}^n \sum_{j \neq i} \int_0^T X_{ij}(\beta, t) w_{ij}(\beta^o, t)\lambda_0(t)\,dt.
	\end{align*}
	Since $M_{ij}(\beta^o, t) = N_{ij}(t) - \int_0^t w_{ij}(\beta^o, s)\lambda_0(s)\,ds$ is mean zero, we have that $X(\beta) - A(\beta)$ is mean zero.  For any fixed $\beta$, it holds that
	\begin{align}
	& \mathrm{var}\bigl\{X(\beta) - A(\beta)\bigr\} = \frac{1}{n^2(n-1)^2}\mathrm{var}\left\{\sum_{i=1}^n\sum_{j \neq i} \int_0^T X_{ij}(\beta, t) \, dM_{ij}(\beta^o, t) \right\} \nonumber \\
	= & \frac{1}{n(n-1)}E\left[\left\{\int_0^T X_{ij}(\beta, t)\, dM_{ij}(\beta^o, t)\right\}\left\{\sum_{(k, l) \in \mathcal{J}_{ij}} \int_0^T X_{kl}(\beta, t)\, dM_{kl}(\beta^o, t)\right\}\right]\nonumber\\
	= & O(n^{-1}(n-1)^{-1}m_n). \label{eq-proof-thm1-1}
	\end{align}
	due to Conditions~\ref{assumption-model} and \ref{assumption-regular}.  Recall that $m_n = o(n^{1/4})$, we have $\mathrm{var}\bigl\{X(\beta) - A(\beta)\bigr\} = o(1)$.  It therefore follows from Markov inequality, $X(\beta)$ converges in probability to the same limit as $A(\beta)$, the unique maximum of which is at $\beta^o$.  Therefore, we have $\hat{\beta}$ converges to $\beta^o$ in probability. 
	
\vskip 3mm
\noindent \textbf{Step 2: Asymptotic normality.}	
We are to derive the asymptotic normality of 
	\begin{eqnarray}
	U_n(\beta^o)   & = & \sum_{i=1}^n\sum_{j\neq i}U_{ij}(\beta^o) =  \sum_{i=1}^n \sum_{j\neq i}\int_0^T \left\{Z_{ij}(t) - \bar{Z}_n(\beta^o, t)\right\}\, dN_{ij}(t) \nonumber\\
	& = & \sum_{i=1}^n \sum_{j\neq i}\int_0^T \left\{Z_{ij}(t) - \bar{Z}_n(\beta^o, t)\right\}\, dM_{ij}(\beta^o, t) \nonumber\\
	& = & \sum_{i=1}^n \sum_{j\neq i}\int_0^T \{Z_{ij}(t) \nonumber \\
	&&- \mu(\beta^o, t)\} \, dM_{ij}(\beta^o, t) - \sum_{i=1}^n \sum_{j\neq i}\int_0^T \{\mu(\beta^o, t) - \bar{Z}_n(\beta^o, t)\}\, dM_{ij}(\beta^o, t). \nonumber
	\end{eqnarray}

Let 
	\begin{equation}\label{eq-mbarvt}
		\bar{M}^v(t) = \frac{v^{\T}\sum_{i=1}^n \sum_{j\neq i} M_{ij}(\beta^o, t)}{\sqrt{\mathrm{var}\left(v^{\T}\sum_{i=1}^n \sum_{j\neq i} M_{ij}(\beta^o, T)v\right)}},
	\end{equation}

For any fixed $v \in \mathbb{R}^p$ and $t\in \mathcal{T}$, let
		\[
		\bar{M}^v_Z(t) = \frac{v^{\T}\sum_{i=1}^n \sum_{j\neq i} \int_{0}^t Z_{ij}(s)\, dM_{ij}(\beta^o, s)}{\sqrt{\mathrm{var}\left(v^{\T}\sum_{j\neq i} \int_{0}^{T} Z_{ij}(s)\, dM_{ij}(\beta^o, s)v\right)}}.
		\]

It follows from Section~A.2 in \cite{LinEtal2000}, it suffices to show the weak convergence of $(\bar{M}^v, \bar{M}^v_Z)$, which means we are to show (i) the finite dimensional distributions converges, and (ii) tightness, see e.g. Theorem~10.2 in \cite{Pollard-1990} or Theorem~1.5.4 of \cite{VaartWellner1996}.

As for (i), it follows from Corollary~2 in \cite{BaldiRinott1989} (see also \citealp{ChenEtal2010}), that we have $(\bar{M}^v, \bar{M}^v_Z)$ converges in finite dimensional distributions to a zero-mean Gaussian process with convergence rate $Q^{-1/2}$, where 
	\begin{equation}
	Q \leq C\frac{nm_n^2}{n^{3/2}} = o(1) \label{eq-proof-1} 
    \end{equation}
	due to Condition~\ref{assumption-model} with a sufficiently large constant $C > 0$.  
	
As for (ii), we are to show for both processes $\bar{M}^v$ and $\bar{M}^v_Z$ satisfy that for each $\varepsilon, \eta > 0$, there is a $\delta > 0$ such that 
	\begin{equation}\label{eq-Mv-tightness}
		\limsup_n \mathrm{pr}\left\{\sup_{|s-t| < \delta} |\bar{M}^v(s) - \bar{M}^v(t)| > \eta \right\} < \varepsilon,
	\end{equation}
	and
	\[
		\limsup_n \mathrm{pr}\left\{\sup_{|s-t| < \delta} |\bar{M}^v_Z(s) - \bar{M}_Z^v(t)| > \eta \right\} < \varepsilon.
	\]
	The above hold by applying Lemma~\ref{lem-1}.

Therefore, we have established the weak convergence of $(\bar{M}^v, \bar{M}^v_Z)$. By similar arguments in Section~A.2 in \cite{LinEtal2000}, one can obtain that	
	\[
	\frac{v^{\T}U_n(\beta^o)}{\left(v^{\T}\Sigma_{2, n}v\right)^{1/2}} \to \mathcal{N}(0, 1), 
	\]
	in distribution as $n \to \infty$.

\vskip 3mm
\noindent \textbf{Step 3: Controlling residuals.}  
We apply the mean value theorem and obtain, for $r = 1, \ldots, p$,
	\begin{equation}\label{eq-mvt}
	\bigl({U}_n(\hat{\beta})\bigr)_r = \bigl({U}_n(\beta^o)\bigr)_r + \biggl(\frac{\partial \bigl({U}_n(\beta)\bigr)_r}{\partial \beta}\bigg |_{\beta = \beta^{(r)}}\biggr)^{\T}\bigl(\hat{\beta} - \beta^o\bigr),	
	\end{equation}
	where $\beta^{(r)}$ lies between $\hat{\beta}$ and $\beta^o$; we therefore have the vector form of Equation~\eqref{eq-mvt} as follows,
	\begin{align}
	{U}_n(\hat{ {\beta}}) & = {U}_n( {\beta}^o) + \left(\begin{array}{c}
			\Bigl(\partial \bigl({U}_n( {\beta})\bigr)_1/\partial  {\beta}\big|_{ {\beta} =  {\beta}^{(1)}}\Bigr)^{\T} \\
			\vdots \\
			\Bigl(\partial \bigl({ {U}_n}( {\beta})\bigr)_p/\partial  {\beta}\big|_{ {\beta} =  {\beta}^{(p)}}\Bigr)^{\T}
			\end{array}
		\right)\bigl(\hat{ {\beta}} -  {\beta}^o\bigr)	\nonumber	\\
	& = { {U}_n}( {\beta}^o) + {M}( {\beta}^*)\bigl(\hat{ {\beta}} -  {\beta}^o\bigr) = 0,\nonumber 
	\end{align}
	where 
	\[
	{M}( {\beta}^*) = \left(\begin{array}{c}
			\Bigl(\partial \bigl({U}_n( {\beta})\bigr)_1/\partial  {\beta}\big|_{ {\beta} =  {\beta}^{(1)}}\Bigr)^{\T} \\
			\vdots \\
			\Bigl(\partial \bigl({ {U}_n}( {\beta})\bigr)_p/\partial  {\beta}\big|_{ {\beta} =  {\beta}^{(p)}}\Bigr)^{\T}
			\end{array}
		\right).
	\]
	Then for any unit-length $v \in \mathbb{R}^p$, 
	\begin{align*}
	v^{\T}\bigl(\hat{ {\beta}} -  {\beta}^o\bigr) = -v^{\T}\Sigma_{1, n}^{-1}{ {U}_n}( {\beta}^o) + v^{\T}\bigl\{I + \Sigma_{1, n}^{-1}{M}\bigl( {\beta}^*\bigr)\bigr\}\bigl(\hat{ {\beta}} -  {\beta}^o\bigr) = -v^{\T}\Sigma_{1, n}^{-1}{ {U}_n}( {\beta}^o) + o_p(1),
	\end{align*}
	where the last identity follows from the definitions of $\Sigma_{1}$ and $M(\beta^*)$ in addition to applying arguments in Proposition 5 in \cite{YuEtal2018} to the fixed $p$ scenario, and the fact that $\hat{\beta} \to \beta^o$ in probability, as $n \to \infty$. 

Combining the aforementioned three steps, we have
	\[
	\frac{{v}^{\T}(\hat{ {\beta}} -  {\beta}^o)}{\left(v^{\T}\Sigma_{1, n}^{-1}\Sigma_{2, n}\Sigma_{1, n}^{-1}v\right)^{1/2}} \stackrel{\mathcal{D}}{\to} \mathcal{N}(0, 1), \quad \mbox{as } n\to \infty.
	\]	
	
\end{proof}

\bibliographystyle{ims}
\bibliography{recurrent}

\end{document}